\documentclass[runningheads]{article}

\usepackage{amsmath,amsfonts,amscd,upref,amstext,amssymb}
\usepackage{amsthm}
\usepackage{tikz}
\usepackage{scalefnt}
\usetikzlibrary{trees}
\usetikzlibrary{positioning,quotes}
\usepackage{cancel}
\usepackage[colorlinks=true,urlcolor=blue, citecolor=red,linkcolor=blue,linktocpage,pdfpagelabels, bookmarksnumbered,bookmarksopen]{hyperref}
\usepackage[hyperpageref]{backref}
\usepackage[noabbrev]{cleveref}
\usepackage{comment}
\usepackage[shortlabels]{enumitem}
\usepackage{mdframed}
\usepackage[margin=1in]{geometry}
\usepackage{algpseudocodex}
\usepackage{algorithm}
\usepackage{comment}

\allowdisplaybreaks

\newcommand{\qedwhite}{\hfill \ensuremath{\Box}}

\newtheorem{theorem}{Theorem}
\newtheorem{lemma}[theorem]{Lemma}

\usepackage{xcolor}  

\usepackage[colorinlistoftodos,prependcaption,textsize=tiny]{todonotes}

\title{An $\Omega(n \log n)$ Randomized Lower Bound for Cutting a Cake into Proportionally Fair Pieces }

\author{Stephen Arndt\thanks{Carnegie Mellon University, Tepper School of Business. }  \and Kirk Pruhs\thanks{University of Pittsburgh. Supported in part by NSF grant CCF-2209654.} \and Trung Tran\thanks{University of Pittsburgh}}

\begin{document}
\maketitle
\begin{abstract}
 We consider the classic cake cutting problem in the Robertson-Webb model, with the objective
 of proportional fairness. We show that any randomized algorithm must use $\Omega(n \log n)$ queries. 
\end{abstract}

\section{Introduction}\label{sec:intro}

\subsection{Problem Description}\label{subsec:problem_description}
In the cake cutting problem, $n$ players must partition a certain resource $\mathcal{C}$, such that each player $p \in [1, n]$ receives a ``fair'' share of $\mathcal{C}$ according to their own private measure $\mu_p$. Consistent with much of the cake cutting literature, we assume the ``cake'' $\mathcal{C}$ is the unit interval $\mathcal{C} = [0, 1]$. Further, the private measures $\mu_p$ are defined on collections of subintervals of $\mathcal{C}$, and satisfy the following conditions.

\begin{enumerate}
    \item \textbf{Definedness:} For all intervals $I \subseteq [0, 1]$,  $\mu_p(I)$ is well-defined.
    \item \textbf{Nonnegativity:} For all intervals $I \subseteq [0, 1]$, $\mu_p(I) \geq 0$.
    \item \textbf{Additivity:} For all pairs of disjoint intervals $I_1, I_2 \subseteq [0, 1]$, $\mu_p(I_1 \cup I_2) = \mu_p(I_1) + \mu_p(I_2)$.
    \item \textbf{Divisibility:} For all intervals $I \subseteq [0, 1]$ and $\lambda \in [0, 1]$, there exists a subinterval $I' \subseteq I$ such that $\mu_p(I') = \lambda \mu_p(I)$.
    \item \textbf{Normalization:} $\mu_p(\mathcal{C}) = 1$.
\end{enumerate}


A cake division protocol is an iterative procedure which queries players about their measures $\mu_p$, until eventually partitioning the cake $\mathcal{C}$ into $n$ pieces and assigning a piece to each player $p \in [1, n]$. These queries are either \textit{cuts}, which require players to identify an interval of a given value $v \in [0, 1]$, or \textit{evaluation queries}, which require players to evaluate a given interval $I$. We study the fairness condition of \textit{proportionality}, where a cake division protocol must guarantee each player $p$ receives a piece of value at least $1/n$ according to their own private measure $\mu_p$. The full description of our model, introduced by Robertson and Webb \cite{robertson-webb} and expanded upon by Woeginger and Sgall \cite{woeginger-sgall}, is given in \Cref{sec:rccm}.

\subsection{Previous Results}\label{subsec:previous_results}
The cake cutting problem arose from Polish mathematicians in the 1940s. Let us begin with upper bounds. Steinhaus \cite{steinhaus} described a deterministic protocol which uses $\Theta(n^2)$ cuts. Even and Paz \cite{even-paz} gave a deterministic divide-and-conquer protocol which uses $\Theta(n \log n)$ cuts, and this remains the best-known deterministic protocol (based on cut / query complexity) to date. A central open question in this field is whether there exists a deterministic protocol which uses $\Theta(n)$ cuts -- Even and Paz \cite{even-paz} explicitly conjecture that no such protocol exists, and Robertson and Webb \cite{robertson-webb} strengthen this conjecture by claiming they would place their money against a substantial improvement on the $\Theta(n \log n)$-cut upper bound. 

We now turn to lower bounds. In 1998, Robertson and Webb \cite{robertson-webb} introduced the influential Robertson-Webb query model, a model of computation for cake cutting which is based on the aforementioned cut and evaluation queries.  Woeginger and Sgall \cite{woeginger-sgall} showed a deterministic $\Omega(n \log n)$ query lower bound on the number of queries (cuts plus evaluations), under the restriction players must receive connected pieces. 
Further in the conclusion of their paper Woeginger and Sgall asked whether this lower bound
can extended to randomized algorithms. 
Later Edmonds and Pruhs \cite{edmonds-pruhs} extended the lower bound of Woeginger and Sgall,
removing the restriction that the assigned pieces needed to be subintervals. 
Further Edmonds and Pruhs claimed without proof that the lower bound on Woeginger
and Sgall could be extended to randomized algorithms.

\subsection{Our Results and Organization of the Paper}\label{subsec:previous_results}

We prove that there is an $\Omega(n \log n)$  lower bound for the expected number of queries used by a  randomized protocol under the restriction that players receive connected pieces, thus substantiating the claim from \cite{edmonds-pruhs}. 

In \Cref{sec:rccm}, we fully define the restricted cake cutting model introduced by Woeginger and Sgall \cite{woeginger-sgall}. In \Cref{sec:lb_construction}, we describe a set of instances introduced by these authors which will be crucial for our result, and then we prove our main result.

\section{The Restricted Cake Cutting Model}\label{sec:rccm}
In \Cref{subsec:rw_model}, we describe the Robertson-Webb Model. In \Cref{subsec:ws_model}, we describe the Woeginger-Sgall Restricted Model. In \Cref{subsec:discussion_models}, we provide a brief discussion on these models.

\subsection{The Robertson-Webb Model \cite{robertson-webb}}\label{subsec:rw_model}
Recall the cake $\mathcal{C}$ is the unit interval $\mathcal{C} = [0, 1]$, and each of the $n$ players has a private measure $\mu_p$ on collections of subintervals of $\mathcal{C}$. For $\alpha \in [0, 1]$, define the \textit{$\alpha$-point} of a player $p$ as the infimum over all numbers $x \in [0, 1]$ such that $\mu_p([0, x]) = \alpha$. In the \textit{Robertson-Webb Model}, the following two queries are allowed.

\begin{itemize}
    \item Cut$(p, \alpha)$ returns the $\alpha$-point of player $p$.
    \item Eval$(p, x)$ returns player $p$'s evaluation of the piece $[0, x]$, meaning $\mu_p([0, x])$. $x$ must be the return value of a previous cut.
\end{itemize}

The protocol can also assign pieces to players during its execution.

\begin{itemize}
    \item Assign$(p, x_i, x_j)$ assigns the interval $[x_i, x_j)$ to player $p$. $x_i$ and $x_j$ must be the return values of previous cuts.
\end{itemize}

\subsection{The Woeginger-Sgall Restricted Model \cite{woeginger-sgall}}\label{subsec:ws_model}
In the \textit{Woeginger-Sgall Restricted Model}, Assign$(p, x_i, x_j)$ is issued only once for each player $p$. In other words, each player $p$ must receive a single subinterval of the cake $\mathcal{C}$.

\subsection{Discussion of Models}\label{subsec:discussion_models}
Although the query restrictions in the Robertson-Webb model are severe, they are crucial for avoiding uninteresting, ``cheating'' protocols, and they are general enough to capture most classic discrete protocols in the literature. As a result, the Robertson-Webb model has received considerable attention. The Woeginger-Sgall single subinterval restriction is also severe, and they remark that it seems to significantly cut down the set of possible protocols. However, most of the known protocols in the literature obey this restriction, including the canonical Even-Paz divide-and-conquer protocols.

\section{The Lower Bound Construction}\label{sec:lb_construction}

\subsection{The Woeginger-Sgall Instances \cite{woeginger-sgall}}
Let $\epsilon < 1/n^4$ be an arbitrary small positive real number. For $i \in [1, n]$, let $X_i \subset [0, 1]$ be the set of $n$ points $i/(n+1) + k\epsilon$ for $k \in [1, n]$. Further, let $X = \bigcup_{1 \leq i \leq n} X_i$. We then design our instances so that for all $i \in [1, n]$, the $i/n$-points of players $p = 1, 2, \dots, n$ are distinct and from the $n$ points $X_i$.

We now describe how we achieve this design. For all players $p \in [1, n]$, all of the cake value for player $p$ is concentrated in tiny intervals $\mathcal{I}_{p, i}$ of length $\epsilon$ around their $i/n$-points, where $i \in [1, n]$. Specifically, player $p$ has cake value $i/(n^2+n)$ immediately to the left of their $i/n$-point, and cake value $(n-i)/(n^2+n)$ immediately to the right of their $i/n$-point. Observe that each player $p$ correctly has cake value $1/n$ between their $(i-1)/n$-point and $i/n$-point for all $i \in [1, n]$, and has cake value $0$ between adjacent intervals $\mathcal{I}_{p, i}$ and $\mathcal{I}_{p, i+1}$.

We may also characterize each instance by a string of length $n^2$, consisting of $n$ ``chunks'' of length $n$ which each contain a (possibly different) permutation of the players $p = 1, 2, \dots, n$. For $i \in [1, n]$, the $i$'th chunk describes the order of the players' $i/n$-points on $X_i$. Thus the set of all such instances is in bijection with the set of all such strings. Moving forward, we will only concern ourselves with this set of instances in the restricted cake cutting model.

Define a \textit{primitive protocol} as a protocol which only makes cuts at $i/n$-points of players $p$ for $i \in [1, n]$.

\begin{lemma}\label{lemma:primitive}\cite{woeginger-sgall}
On the set of instances $I$ in the restricted model described above, the following holds. For all fair protocols $\mathcal{P}$ on instance $I$, there exists a fair primitive protocol $\mathcal{P}'$ on instance $I$, which makes the same number of cuts and evaluation queries.
\end{lemma}

Thus a lower bound on the query complexity of \textit{primitive} protocols $\mathcal{P}^*$, on the set of instances $I$ described above, will induce a lower bound on the query complexity of all protocols $\mathcal{P}$. In fact, we may extend \Cref{lemma:primitive} to \textit{cuts-only} primitive protocols $\mathcal{P}'$.

\begin{lemma}\label{lemma:cuts_only_primitive}
On the set of instances $I$ in the restricted model described above, the following holds. For all fair primitive protocols $\mathcal{P}^*$ on instance $I$, there exists a fair primitive protocol $\mathcal{P}'$ on instance $I$ which only uses cuts, and where the number of cuts of $\mathcal{P}'$ equals the number of cuts and evaluation queries of $\mathcal{P}^*$.
\end{lemma}

\begin{proof}
$\mathcal{P}'$ simulates $\mathcal{P}^*$. If $\mathcal{P}^*$ calls Cut$(p, i/n)$, $\mathcal{P}'$ calls Cut$(p, i/n)$. If $\mathcal{P}^*$ calls Eval$(p, x)$ where $x \in X_i$ is the result of some previous Cut$(p', i/n)$, $\mathcal{P}'$ calls Cut$(p, i/n)$. Observe that from the evaluation query, $\mathcal{P}^*$ only learns whether $p$'s $i/n$-point is to the left, right, or in the same place (if $p' = p$) as $p'$'s $i/n$-point. $\mathcal{P}'$ performs Cut$(p, i/n)$, which reveals the exact location of $p$'s $i/n$-point, and thus its position relative to $p'$'s $i/n$-point. Thus $\mathcal{P}'$ can properly simulate $\mathcal{P}^*$ under this replacement of $\mathcal{P}^*$'s Eval operations. $\qedwhite$
\end{proof}

Thus a lower bound on the cut complexity of \textit{cuts-only} primitive protocols $\mathcal{P}'$, on the set of instances $I$ described above, will induce a lower bound on the query complexity of primitive protocols $\mathcal{P}^*$, and thus a lower bound on the query complexity of all protocols $\mathcal{P}$.

We now characterize the fairness condition of primitive protocols $\mathcal{P}^*$. Upon termination, $\mathcal{P}^*$ must identify $n-1$ cuts performed, say at positions $0 \leq y_1 \leq y_2 \leq \dots \leq y_{n-1} \leq 1$; further define $y_0 = 0$ and $y_1 = 1$. $\mathcal{P}^*$ then must assign the $n$ pieces $[y_{i-1}, y_i]$ for $1 \leq i \leq n$ according to the permutation of players $\phi$. On a subset of these instances $I$, we can enforce an explicit fairness condition.

Consider the following subset of instances $J$. Let $\pi$ be an arbitrary permutation of the integers $1, 2, \dots, n$. For all $i \in [1, n]$, and players $p \in [1, n]$,  

\begin{itemize}
    \item If $\pi(p) = i$, let player $p$'s $i/n$-point be the $i$'th point in $X_i$.
    \item If $\pi(p) < i$, require player $p$'s $i/n$-point to be before the $i$'th point in $X_i$.
    \item If $\pi(p) > i$, require player $p$'s $i/n$-point to be after the $i$'th point in $X_i$.
\end{itemize}

\begin{lemma}\label{lemma:must_find_pi}\cite{woeginger-sgall}
If the primitive protocol $\mathcal{P}^*$ is fair on an instance $J$ from this subset with underlying permutation $\pi$, then $\phi = \pi^{-1}$.
\end{lemma}

We now have a clear understanding of the class of protocols we must lower bound (cuts-only primitive protocols), and a required condition for achieving fairness (identifying the underlying permutation $\pi$ of an instance $J$ from this subset).

\subsection{Decision Tree Argument}
Let $\mathcal{D}$ be a uniform distribution over the subset of instances $J$ described above. We now state our main theorem, \Cref{thm:main}.

\begin{theorem}\label{thm:main}
All fair deterministic cuts-only primitive protocols require $\Omega(n \log n)$ cuts and evaluation queries on expectation on input distribution $\mathcal{D}$.
\end{theorem}

We will prove \Cref{thm:main} using a decision tree argument. Via Yao's Principle, \Cref{thm:main}, \Cref{lemma:primitive}, and \Cref{lemma:cuts_only_primitive} induce an $\Omega(n\log n)$ lower bound on the query complexity of all randomized protocols $\mathcal{P}$. First, we describe a randomized adversary strategy which induces this distribution $\mathcal{D}$, to simplify our argumentation.

Consider the following adversary strategy. Let $\mathcal{P}'$ be an arbitrary cuts-only primitive protocol. Select a permutation $\pi$ over the integers $1, 2, \dots, n$ uniformly at random. On each Cut$(p, i/n)$, check if $\pi(p) < i$, $\pi(p) = i$, or $\pi(p) > i$. Uniformly place $p$'s $i/n$-point among the valid remaining $i/n$-points, obeying the requirements of the instances $J$. If $\mathcal{P}'$ terminates, uniformly place the remaining $i/n$-points of the players according to these requirements.

\begin{lemma}\label{lemma:adv_strategy}
The described adversary strategy generates an instance according to uniform distribution $\mathcal{D}$.
\end{lemma}

\begin{proof}
The uniform choice of $\pi$ is correct, because $\mathcal{D}$ has an equal number of instances with nonzero probability associated with any given permutation. Further, the $i/n$-points of all remaining players are chosen uniformly at random by the adversary strategy, and so it correctly produces a uniform sample over the instances with nonzero probability in $\mathcal{D}$.
\end{proof}

We now proceed with the decision tree argument. Construct the decision tree $T = (V, E)$ of $\mathcal{P}'$. Each node of $T$ represents a state of $\mathcal{P}'$, and contains the next cut operation. The edges down the tree $T$ contain the answers to these cut queries, and thus the state transitions of $\mathcal{P}'$. Observe that the branching factor of $T$ is variable and up to size $n$. This high branching factor means a naive decision tree argument will not induce the correct lower bound. Instead, we must prove some important properties of $T$.

\begin{lemma}\label{lemma:uniform_pis}
Consider an arbitrary node $v \in V$ in the decision tree. Let $\Pi$ be the set of valid permutations $\pi$ at $v$, meaning those permutations $\pi$ which could still describe the instance at $v$. Then under input distribution $\mathcal{D}$, $\Pi$ also has a uniform distribution.
\end{lemma}

\begin{proof}
Let $\pi_1, \pi_2 \in \Pi$ be arbitrary. The described adversary strategy, which is equivalent to sampling from distribution $\mathcal{D}$ via \Cref{lemma:adv_strategy}, chooses its instance identically under $\pi_1$ or $\pi_2$. Thus the current distribution of ``incomplete'' instances at node $v$ is identical under $\pi_1$ or $\pi_2$, so the probability $\pi_1$ or $\pi_2$ is the correct permutation is equal. $\qedwhite$
\end{proof}

\begin{lemma}\label{lemma:expected_depth}
Let $v \in V$ be arbitrary. Suppose there are $N$ valid permutations $\pi$ at $v$. Then the expected depth of $v$ in $T$, under input distribution $\mathcal{D}$, is at least $\log_3 N + 1$.
\end{lemma}

\begin{proof}
We proceed by induction. Via \Cref{lemma:must_find_pi}, the leaves must have $N = 1$, and so the base case holds. Next, let $v \in V$ be an arbitrary non-leaf node in $T$ with $N$ valid permutations $\pi$. Suppose Cut$(p, i/n)$ is performed at $v$. Let $v_<$, $v_=$, and $v_>$ be the sets of children of $v$ with $\pi(p) < i$, $\pi(p) = i$, and $\pi(p) > i$, respectively. Suppose the nodes of $v_<$, $v_=$, and $v_>$ have $aN, bN$, and $cN$ valid permutations $\pi$ where $a + b + c = 1$. Because the set of valid permutations $\pi$ at $v$ has a uniform distribution via \Cref{lemma:uniform_pis}, the probability of transitioning to a node in $v_<, v_=$, or $v_>$ is $a, b$ or $c$. Via the induction, the expected depth of each node in $v_<, v_=$, and $v_>$ is $\log_3 (aN) + 1$, $\log_3 (bN) + 1$, and $\log_3 (cN) + 1$. Thus the expected depth of $v$ is

\begin{align*}
&\geq a\left(\log_3(aN) + 1\right) + b\left(\log_3(bN) + 1\right) + c\left(\log_3(cN) + 1\right) + 1 \\
&= a\log_3(aN) + b\log_3(bN) + c\log_3(cN) + (a+b+c) + 1 \\
&= \left(a\log_3 N + b\log_3 N + c \log_3 N\right) + \left(a\log_3a + b\log_3b + c\log_3 c\right) + 2 \\
&= \log_3 N + \left(a\log_3a + b\log_3b + c\log_3 c\right) + 2
\end{align*}
\end{proof}

Observe $f(x) = x \log_3 x$ is a convex function. Thus via Jensen's inequality,

\begin{align*}
\frac{1}{3}a\log_3a + \frac{1}{3}b\log_3b + \frac{1}{3}c\log_3 c &\geq \left(\frac{1}{3}a + \frac{1}{3}b + \frac{1}{3}c\right) \log_3 \left(\frac{1}{3}a + \frac{1}{3}b + \frac{1}{3}c\right) \\
&= \frac{1}{3}\log_3\frac{1}{3} \\
&= -\frac{1}{3}
\end{align*}

giving $a\log_3a + b\log_3b + c\log_3 c \geq -1$. Substituting this inequality gives the desired result. $\qedwhite$

\begin{proof}(of \Cref{thm:main})
At the root node $r \in V$ of the decision tree $T$, there are $N = n!$ valid permutations $\pi$. Via \Cref{lemma:expected_depth}, the expected depth of $r$, under input distribution $\mathcal{D}$, is at least $\log_3 N + 1 = \log_3 n! + 1 = \Omega(n \log n)$. This is equal to the expected runtime of $\mathcal{P}'$ under input distribution $\mathcal{D}$, completing the proof. $\qedwhite$
\end{proof}

\bibliographystyle{splncs04}
\bibliography{bib}

\end{document}